\documentclass{article}

\usepackage{moreverb}

\usepackage{url}

\usepackage{graphicx}
\usepackage{amsmath}
\usepackage{amsfonts}
\usepackage{multirow}
\usepackage{amssymb}

\newtheorem{corollary}{Corollary}

\newtheorem{theorem}{Theorem}
\newtheorem{example}{Example}
\newtheorem{remark}{Remark}
\newenvironment{proof}[1][Proof]{\textbf{#1.} }{\ \rule{0.5em}{0.5em}}

\newcommand{\sdwell}{{\mathcal S_\text{min}}}
\newcommand{\saverage}{{\mathcal S_\text{ave}}}
\newcommand{\sddwell}{{\mathcal S_{\mathcal G,\text{min}}}}
\newcommand{\sdaverage}{{\mathcal S_{\mathcal G,\text{ave}}}}

\begin{document}

%\runninghead{F.~\.Ilhan and \"O.~Karabacak: Dwell time computation methods}

%\categorytitle{Regular Paper}

\title{GRAPH-BASED DWELL TIME COMPUTATION METHODS FOR DISCRETE-TIME SWITCHED LINEAR SYSTEMS}

\author{F.~\.Ilhan and \"O.~Karabacak}

%\address{The authors are with Istanbul Technical University, Faculty of Electrical and Electronic Engineering, Electronics and Communication Engineering Department, Maslak, TR-34469, Istanbul, Turkey. (karabacak@itu.edu.tr)}

%\received{March 6, 2015.}

%\acks{}

\date{12 September 2015}

\maketitle

\begin{abstract}
Analytical computation methods are proposed for {evaluating} the minimum dwell time and average dwell time {guaranteeing the asymptotic stability} of a discrete-time switched linear system whose switchings are assumed to respect a given directed graph. The minimum and average dwell time  can be found using {the graph that governs the switchings, and the associated weights}. This approach, which is used in a previous work for continuous-time systems having non-defective subsystems, has been adapted to discrete-time switched systems and generalized to allow defective subsystems. Moreover, we present a method to improve the dwell time estimation in the case of bimodal switched systems. In this method, scaling algorithms to minimize the condition number are used to give better minimum dwell time and average dwell time estimates.

\end{abstract}

\section{Introduction}
Problems {regarding} the stability of switched linear systems have been attracted {the interest of}  many researchers in the last two decades \cite{Hespanha_itac_04}. A well-known approach to the stability of switched linear systems is to impose constraints on the set of switching signals so as to guarantee the stability \cite{Morse_itac_96,Geromel_sjco_06,Guo_2012,Cai_2012}. Mainly, 
two kinds of constraints {have been} considered: minimum dwell time constraint and average dwell time constraint. {For the former}, intervals between two consecutive switchings are assumed to be larger than or equal to a number called minimum dwell time, whereas {for the latter}, these intervals are assumed to be{, on average,} larger than or equal to a number called average dwell time. 

In {the} literature, there are many methods to find {the} minimum dwell time and/or {the} average dwell time{ that guarantee stability of} a given switched system  \cite{Shorten2007}. However, in general, there is no method that gives the smallest possible minimum (or average) dwell time in terms of subsystem properties. The most efficient methods that can be used to approximate the minimum (or average) dwell time are based on linear matrix inequalities \cite{Geromel_sjco_06,Chesi2012} and therefore, {they} do not give any insight in the relationship between the subsystem properties and the minimum (or average) dwell time \cite{Haimovich2011}.

{For continuous-time switched systems with non-defective subsystem matrices, an estimate of the} minimum (or average) dwell time explicitly depending on the subsystem properties has been {derived} in \cite{Karabacak_ijss_09,karabacak_scl_13}, where {the} minimum dwell time is found as a function of {the} subsystem eigenvalues and eigenvectors. {The method in \cite{Karabacak_ijss_09} is based on viewing a switching signal as a walk in the complete directed graph (digraph) whose nodes correspond to subsystems and whose arcs correspond to switching events.} 

{In \cite{karabacak_scl_13}, a more general problem has been considered, namely finding the minimum dwell time for switched systems whose switchings are governed by a digraph (called switching graph). Note that the special case of this problem where the switching graph is complete corresponds to the standard dwell time problem in the literature. On the other hand, the general problem is important because} switched systems {whose} switchings are governed by a digraph can be {encountered} in control engineering applications \cite{Hou_jqcd_10,Xu_12,Zhang_12}. Theoretically, such systems are first considered in the
switched system literature, {to the best of our knowledge}, in \cite{Mancilla_05},
where stability conditions are reduced to the conditions on
{the} strongly connected components of the digraph.  {Later on}, the second author of this paper presents sufficient conditions {for} the stability of constrained switched systems using the properties of the digraph that governs the switchings \cite{karabacak_scl_13}. This method has been improved in {a conference paper} \cite{ilhan_14} {for continuous-time systems} so as to cover systems with defective subsystem matrices. Recently, in \cite{Kundu_14}, stabilization of switched systems {whose} switchings are governed by a digraph has been considered. Digraphs are also{ used in \cite{jungers}, where the multiple Lyapunov function method has been generalized}.

% 
% along ****** a specific switching cycle (namely, the cycle with maximum cycle ratio, to be explained soon). This method is based on considering a switched system as a doubly-weighted digraph (called switching graph) with specific weights that provide upper bounds on the norm of any solution. A switching signal is then viewed as a walk in the switching graph. Since any walk can be decomposed into cycles and a path of bounded length, conditions can be imposed on the cycles only.******** The so-called maximum cycle ratio of the switching graph then provides a minimum dwell time estimate, whereas the so-called maximum cycle mean of the switching 
% graph provides an average dwell time 
% estimate. \textbf{This method applies to 
%  switched system whose switchings are governed by a digraph.} For, in this case, already existing constraints on 
% the switchings have been naturally considered by the switching graph, giving rise to much smaller dwell-time estimates.
% 
% Switched systems \textbf{whose} switchings are governed by a digraph can be \textbf{encountered} in control engineering applications \cite{Hou_jqcd_10,Xu_12,Zhang_12}. Theoretically, such systems are first considered in the
% switched system literature, \textbf{to the best of our knowledge}, in \cite{Mancilla_05},
% where stability conditions are reduced to the conditions on
% \textbf{the} strongly connected components of the digraph. 

In this paper, we apply the method {proposed} in \cite{karabacak_scl_13} to discrete-time switched linear systems and waive the non-defectiveness condition by considering the Jordan form of subsystem matrices. In addition, we improve the minimum dwell time estimate for bimodal systems by applying a scaling algorithm that minimizes {the} condition number of the matrices. 

In the sequel, we first explain how the switching graph arises naturally considering a bound on the norm of {the solution of a linear switched system}. In Section~\ref{sec:switchinggraph}, we define the switching graph which will be used to estimate {the} minimum (or average) dwell time. Based on the switching graph, methods for minimum dwell time and average dwell time computation are given in Section~\ref{sec:dwell} and Section~\ref{sec:averagedwell}, respectively. Finally, we will discuss on possible future research in this area in Section~\ref{sec:conc}.

{\textbf{Notation:} $\mathbb R$ and $\mathbb N$ denote the set of real numbers and the set of positive integers, respectively. $\|\cdot\|$ denotes the 2-norm for vectors and the spectral norm for matrices.}

\section{Switched Systems and the Switching Graph}
\label{sec:switchinggraph}

We consider discrete-time switched linear systems of the form
\begin{equation}
\label{eq:basic_form}
x(t+1)=A_{\sigma(t)}x(t),\quad 
\sigma \in \mathcal S ,\quad t \in \mathbb{N},
\end{equation}
where $x\in\mathbb R^N$ is the state of the system, $\{A_i \in \mathbb{R}^{N\times N} \}_{i \in \{1,\dots,M\}}$ is a finite set of Schur stable subsystem matrices, $M$ is the number of subsystems, $N$ {is the dimension of the state space}, $\{1,\dots,M\}$ is the index set for {the} subsystems and $\mathcal S$ denotes the set of admissible switching signals. {Assume as} switching instants $0=t_0<t_1<\dots$, where $t_k \in \mathbb N$ and denote {the index of} the active subsystem{ from time} $t_k$ {to time} $t_{k+1}-1$ by $\sigma_k$. Let $N_\sigma(t)$ denote the number of switchings {in the time interval  $[0,t)$}. Two different sets of switching signals are considered:
\begin{equation}
\label{eq:dwellset}
\sdwell[\tau]=\{\sigma | t_k-t_{k-1}\geq\tau\}
\end{equation}
\begin{equation}
\label{eq:avedwellset}
\saverage[\tau,N_0]=\{\sigma | N_\sigma(t)\leq N_0+\frac{t}{\tau}\}
\end{equation}
$\sdwell[\tau]$ consists of the switching signals {for which} the {interval between two consecutive switchings is} always larger than $\tau$, whereas $\saverage[\tau,N_0]$ is the set of switching signals with the property that at each time $t$ the number of past switchings $N_\sigma(t)$ satisfies the average dwell time condition $N_\sigma(t)\leq N_0+\frac{t}{\tau}$.

For a given initial condition $x(0)$, the solution of the discrete-time switched linear system (\ref{eq:basic_form}){ for }$t \in \{t_n,\dots,t_{n+1}-1\}$ can be written as 
\begin{equation}
\label{eq:main_sol}
x(t)=A_{\sigma_{n+1}}^{(t-t_n)} \left(\prod_{k=1}^{n} A_{\sigma_{k}}^{(t_k-t_{k-1})}\right)x(0).
\end{equation}

Let us consider the Jordan matrix decomposition $A^{}_i = P^{}_iJ^{}_iP_i^{-1}$ where $P_i$ is the generalized eigenvector matrix of $A_i$ and $J_i$ is the Jordan matrix. We use the fact that the $1$'s above the diagonal in the Jordan matrix are conventional and can be replaced by any constant $\epsilon$. To explain this, let us assume that the Jordan form $J$ consists of one 
Jordan block, matrices with a Jordan form that consists of more than one Jordan block can be treated similarly. 
Then, {by changing} the
generalized eigenvector matrix from $P=[p_0|p_1|p_2|\dots]$ to $%
P^{(\epsilon)}=[p_0| \epsilon p_1 | \epsilon^2 p_2 |\dots]$, the matrix $A$ can be written as 
\begin{equation}  
\label{eq:genJordan}
A=P^{(\epsilon)} \cdot J^{(\epsilon)} \cdot {P^{(\epsilon)}}^{-1},
\end{equation}
where $J^{(\epsilon)}$ is in Jordan form with $\epsilon$'s in place of $1$'s. It is known that $\|J^{(\epsilon)}\|\leq\|D\|+\|N^{(\epsilon)}\|$ where $\|D\|$ denotes the diagonal part of $J^{(\epsilon)}$ and $N^{(\epsilon)}$ denotes {the} nilpotent part of $J^{(\epsilon)}$. Then, $\|N^{(\epsilon)}\|=\epsilon$.{ For a Schur stable matrix} $A$, $\|D\|<1$,{ and we can choose an } $\epsilon$ such that $0<\epsilon<1-\|D\|$. {Then, one gets}  $\|J^{(\epsilon)}\|<1$. {For a matrix whose Jordan form consists of more than one Jordan blocks $J_1,J_2,\dots$, let us choose a sufficiently small $\epsilon$ satisfying 
\begin{equation}
\label{eq:1}
\epsilon<1-\|D_k\|
\end{equation}
for each $k$, where $D_k$ is the diagonal part of the Jordan block $J_k$. Since $\|J^{(\epsilon)}\|=\max_k \|J^{(\epsilon)}_k\|$, we obtain 
\begin{equation}
 \label{ineq:kucukbirden}
 \|J^{(\epsilon)}\|<1.
\end{equation}
 Moreover, since $\|D\|=\max_k \|D_k\|$, the condition (\ref{eq:1}) is equivalent to}
 \begin{equation}
  \label{eq:2}
  \epsilon<1-\|D\|.
 \end{equation}
 
{For each subsystem matrix} $A_i$, {we choose a generalized Jordan decomposition} $A^{}_i=P^{(\epsilon_i)}_i J^{(\epsilon_i)}_i {P^{(\epsilon_i)}_i}^{-1}$ with a sufficiently small $\epsilon_i$ satisfying 
\begin{equation}
 \label{eq:3}
\epsilon_i<1-\|D_i\|,
\end{equation}
 where $D_i$ is the diagonal part of $J_i$. Then, we have $
 \|J_i^{(\epsilon_i)}\|<1.$  {In the following we drop the superscript $(\epsilon)$ from matrices $P^{(\epsilon)}$ and $J^{(\epsilon)}$ for the simplicity of notation, unless it is necessary to indicate the dependence on $\epsilon$. Note that for a non-defective subsystem matrix $A_i$, the Jordan form is diagonal and the Schur stability implies that $\|J_i\|=\|D_i\|<1$.}
 
Substituting $A^{}_i=P^{}_i J^{}_i {P^{}_i}^{-1}$ in (\ref{eq:main_sol}), and using norm inequalities we have
\begin{multline}
\|x(t)\|=\left\|A_{\sigma_{n+1}}^{t-t_n} \left(\prod_{k=1}^{n} A_{\sigma_{k}}^{(t_k-t_{k-1})}\right)x(0)\right\| \\
=\left\|   P_{\sigma_{n+1}}^{}J_{\sigma_{n+1}}^{(t-t_n)}  \left(\prod_{k=1}^{n}P_{\sigma_{k+1}}^{-1}P_{\sigma_k}^{} J_{\sigma_{k}}^{(t_k-t_{k-1})}\right)P_{\sigma_1}^{-1}x(0)            \right\|\\
\label{eq:same_as_jordan}
\leq \|P_{\sigma_{n+1}}^{}\|\|P_{\sigma_{1}}^{-1}\| \|J_{\sigma_{n+1}}^{}\|^{(t-t_n)}\cdot \\ \left( \prod_{k=1}^{n} \|P_{\sigma_{k+1}}^{-1}P_{\sigma_k}^{}\|\|J_{\sigma_k}^{}\|^{(t_k-t_{k-1})}   \right)\cdot \|x(0)\|,
\end{multline}
%{where we drop the superscript $(\epsilon)$ from matrices $P^{(\epsilon)}$ and $J^{(\epsilon)}$ for simplicity of notation.}

Let us define $\rho_i^{(\epsilon_i)}:=\|J_i^{(\epsilon_i)}\|$. {Note that, if $A_i$ is non-defective, then $J_i$ is diagonal, and hence $\rho_i$ is equal to the spectral radius of $A_i$.} Then, writing the term{ between parentheses in Inequality (\ref{eq:same_as_jordan})} in exponential form, we have
\begin{multline}
\|x(t)\|\leq\|P_{\sigma_{n+1}}^{}\|\|P_{\sigma_{1}}^{-1}\| \rho_{\sigma_{n+1}}^{(t-t_n)} \cdot \\ e^{\sum_{k=1}^{n}\left(  \ln\|P_{\sigma_{k+1}}^{-1}P_{\sigma_k}^{}\|+(t_k-t_{k-1})\ln\rho_{\sigma_k}   \right)} \|x(0)\|\\
\leq\gamma \rho_{\sigma_{n+1}}^{(t-t_n)} e^{ \sum_{k=1}^{n}\left( \ln\|P_{\sigma_{k+1}}^{-1}P_{\sigma_k}^{}\|+(t_k-t_{k-1})\ln\rho_{\sigma_k}   \right)} \|x(0)\|,
\label{eq:average_git}
\end{multline}
where $\gamma=\max_{i,j \in \{1,\dots,M\}} \|P_i^{}\|\|P_j^{-1}\| $. Since each subsystem is Schur stable, from (\ref{ineq:kucukbirden}) we have
\begin{eqnarray}
 \label{eq:lnrho}
\ln\rho_i=\ln\|J_i\|<0,
\end{eqnarray}
for all $i$. Then, using   $\rho_{\sigma_{n+1}}^{(t-t_n)}<1$ and $\tau\leq t_k-t_{k-1}$ for $\sigma\in\sdwell[\tau]$, we can write
\begin{multline}
\|x(t)\| \leq\gamma^{} e^{ \sum_{k=1}^{n}\left( \ln\|P_{\sigma_{k+1}}^{-1}P_{\sigma_k}^{}\|+\tau\ln\rho_{\sigma_k}   \right)} \|x(0)\|
\\
=\gamma e^{\alpha(n)}\|x(0)\|,
\end{multline}
 where 
 \begin{equation}
 \label{eq:alpha_function}
 \alpha(n)=\sum_{k=1}^{n} \ln\|P_{\sigma_{k+1}}^{-1}P_{\sigma_k}^{}\|+\tau\ln\rho_{\sigma_k}.
 \end{equation}
In the following, we show that the function $\alpha(n)$ can be seen as the weight of a walk (of length $n$) on a doubly weighted digraph called the switching graph. 

\subsection{Definition of a switching graph}
For some switched systems, transitions between subsystems are restricted by a digraph,{ whose } nodes represent subsystems and {whose} directed edges represent admissible transitions between subsystems. As a consequence of this idea, a switching signal $\sigma$ can be viewed as a walk on such a graph. It can be easily seen that each transition from subsystem $i$ to subsystem $j$ {gives a contribution to} $\alpha$ function which is a function of the ordered pair $(i,j)$, namely $\ln\|{(P^{(\epsilon_j)}_{j})}^{-1}P^{(\epsilon_i)}_{i}\|+\tau\ln\rho_{i}^{\epsilon_i}$. These values can be assigned as {the} weights of the directed edges for all admissible transitions between subsystems. Consider $\omega_{ij}^+=\ln\|{(P_{j}^{(\epsilon_j)})}^{-1}P_i^{(\epsilon_i)}\|$ as the gain of the transition from subsystem $i$ to subsystem $j$  and $\omega_{ij}^-=-\ln\rho_{i}^{(\epsilon_i)} $ as{ the loss of the dwelling} at subsystem $i$ {per unit time}. Then, we have a doubly weighted graph as follows.

A \textit{switching graph} of a switched linear system (\ref{eq:basic_form}) is a doubly weighted digraph 
\begin{equation}
\mathcal G = \{ \mathcal V, \mathcal E, \omega^+, \omega^-\} .
\end{equation}
Here, $\mathcal V$ is the set of nodes which is isomorphic to the index set of subsystems. $\mathcal E$ is the set of the directed edges that represent admissible transitions between subsystems. This set is given by  $\mathcal E = \{(i,j) | i \neq j, i, j \in \mathcal V\}$ in the case of no restriction imposed on transitions between subsystems, namely we have a fully connected switching graph. $\omega^+$ and $\omega^-$ {are the real-valued weight functions defined on the set} $\mathcal E$ as
\begin{equation}
\label{eq:omega_plus}
\omega_{ij}^+=\ln\|{(P_{j}^{(\epsilon_j)})}^{-1}P_i^{(\epsilon_i)}\|\end{equation}
\begin{equation}
\label{eq:omega_minus}
\omega_{ij}^-=-\ln\rho_{i}^{(\epsilon_i)} \end{equation}

For a switched system with four subsystems, the switching graph is shown in Figure~\ref{fig:fully}.
\begin{figure}
\centering
\includegraphics[width=0.4\textwidth]{./fully_new.ps}
\caption{The switching graph ($\mathcal G_1$) of a switched system with four subsystems. On each edge, $\omega^+$ and $\omega^-$ are shown, respectively.}
\label{fig:fully}
\end{figure}
% 
% $\omega^+$ is well-defined when the eigenvector matrix is composed of the eigenvectors with unitary Euclidean norm and in the case of $k$-multiple eigenvalues, eigenvectors are chosen as any $k$ orthogonal vectors that span the corresponding eigenspace. In this case, the eigenvector matrix is well-defined up to a multiplication from right by a unitary matrix. Consider the eigenvector matrices $P_i$ and $P_j$. Choose different eigenvector matrices as $P_i'=P_i \cdot U_i$ and $P_j'=P_j \cdot U_j$. 
% \begin{equation}
% \omega_{ij}^+=\ln\|P_{j}'^{-1}P_{i}'\|=\ln\|U_{}^{-1}P_{j}^{-1}P_{i}^{}U\|=\ln\|P_{j}^{-1}P_{i}^{}\|
% \end{equation}
% where the final equation comes from the fact that the spectral norm is unitarily invariant.
% Note that choosing eigenvectors with a different Euclidean norm can change the value of $\omega_{ij}^+$. This case is discussed in Section~\ref{sub:bimodal_dwell} and Section~\ref{sub:bimodal_average}.

Assume that the switchings in the switched system under consideration should respect a directed graph. In this case, the set of admissible switching signals $\mathcal S$ is restricted by the directed edges of the switching graph. We denote {by}
\begin{equation}
\sddwell[\tau]=\{\sigma \in \sdwell[\tau] | (\sigma_k, \sigma_{k+1}) \in \mathcal E \}
\end{equation}
{the set} $\sdwell[\tau]$ {restricted by the switching graph} $\mathcal G$.
This set of switching signals contains the switching signals that respect the given digraph $\mathcal G$ {and satisfies} the minimum dwell time property. Similarly, we have
\begin{equation}
\sdaverage[\tau,N_0]=\{\sigma \in \saverage[\tau,N_0] | (\sigma_k, \sigma_{k+1}) \in \mathcal E \}.
\end{equation}
\subsection{Maximum cycle ratio}
Consider a weighted digraph $\mathcal G=\{ \mathcal V, \mathcal E, \omega\}$, where $\omega$ is a real-valued weight function defined on $\mathcal E$. A walk $\mathcal W$ {of length $s$} can be defined as $\mathcal W=\left((p_1,p_2),(p_2,p_3),\dots,(p_s,p_{s+1})\right)$ where $(p_i,p_{i+1})\in \mathcal E$ for all $i=1,\dots,s$. A path is a walk $((p_1,p_2),(p_2,p_3),\dots,(p_s,p_{s+1}))$ where $p_1,\dots,p_s$ are distinct and a cycle is a path with $p _{s+1}=p_1$. 
The weight of a walk $\mathcal W$ is defined as $\omega(\mathcal W)= \sum_{k=1}^{s}\omega{((p_k,p_{k+1}))}$. 

Now consider a doubly weighted graph $\mathcal G=\{\mathcal V, \mathcal E, \omega^+, \omega^-\}$. The ratio of a cycle $C$ {in} $\mathcal G$ is defined as $\nu(C)= \frac{\omega^+(C)}{\omega^-(C)}$. The maximum cycle ratio $\nu $ is defined as 
\begin{eqnarray}
\label{eq:max_cycle_ratio}
\nu (\mathcal G)&=& \max_{C \in \mathcal C}\nu(C)=\max_{C \in \mathcal C}\frac{\omega^+(C)}{\omega^-(C)}\nonumber\\
&=&\max_{C \in \mathcal C}\frac{\sum_{k=1}^{|C|}\omega^+((p_k,p_{k+1}))}{\sum_{k=1}^{|C|}\omega^-((p_k,p_{k+1}))}
\end{eqnarray} 
{where $\mathcal C $ denotes the set of all cycles {in} $\mathcal G$ and $|C|$ is the length of the cycle $C=(p_1,\dots,p_{|C|})$.
}

Similarly, the mean of a cycle is defined as $\mu(C)=\frac{\omega^+ (C)}{|C|}$. The maximum cycle mean $\mu $ is defined as
\begin{equation}
\mu (\mathcal G)=\max_{C \in \mathcal C}\mu(C)=\max_{C \in \mathcal C}\frac{\omega^+ (C)}{|C|}.
\end{equation}

{The} optimum (minimum or maximum) cycle ratio, also known as profit-to-time ratio, and {the} optimum cycle mean have been considered in
graph theory literature \cite{Karp_78,Golitschek_82,Dasdan_98,Dasdan_04}, and 
{have} many applications in different areas such as scheduling problems \cite{bouyer_04,Leus_07,Kats_11} and performance analysis of digital systems \cite{Lu_06}. There are many algorithms that can be used to find the optimum cycle ratio and {the} optimum cycle mean for a given doubly weighted digraph (See \cite{Dasdan_04}). In terms of practical complexity, one of the fastest algorithms is given in
\cite{Dasdan_98}. In the sequel, we use this algorithm for which a C code is available in Ali Dasdan's
personal web page \cite{Dasdan_code}.

\section{Minimum Dwell Time Computation}
\label{sec:dwell}
% In this section, we show that an estimate for \textbf{the} minimum dwell time can be given as the maximum cycle ratio of the switching graph. Three different cases \textbf{when} subsystems matrices are non-defective, defective and \textbf{when} the switched system consists of two subsystems (bimodal case) are considered. 

% \subsection{Non-defective case}
{In this section, we consider the switched linear system whose switchings are governed by the digraph} $\mathcal G$ {and the time interval between any consecutive switching instants is larger than or equal to the minimum dwell time} $\tau$. {As a special case, the non-defective bimodal case is discussed in Subsection \ref{sub:bimodal_dwell}.}
  
 \begin{theorem}
 \label{th:nondefective_dwell}
 Let $\{A_i\}_{i=1,\dots,M}$ be a family of Schur stable matrices and {let} $\mathcal G$ be a switching graph. 
 Then the switched linear system given by
 \begin{equation}
\label{eq:nondefective_form}
x(t+1)=A_{\sigma(t)}x(t), \quad 
\sigma \in \sddwell[\tau] ,\quad t \in \mathbb{N}
\end{equation}
is asymptotically stable if
 \begin{equation}
% \label{eq:max_ratio_bound}
 \tau > \nu (\mathcal G),
 \end{equation}
 where the maximum cycle ratio $\nu$ is found using the weights given in (\ref{eq:omega_plus}) and (\ref{eq:omega_minus}) {with parameters} $\epsilon_i$ satisfying (\ref{eq:3}).
 
 \end{theorem} 
 \begin{proof}
   Note that a switching signal $\sigma$ can be represented by a walk $\mathcal W$ {in} the switching graph. If the length of the walk is finite, last subsystem stays active forever {thus guaranteeing} the asymptotic stability of the switched linear system. Hence, we consider walks with infinite length, which represent switching signals having infinitely many switchings. Using Eq. (\ref{eq:alpha_function}), it is seen that $\alpha(n)$ is the weight of the walk
 \begin{equation}
 \mathcal W_n := (\sigma_1, \sigma_2), (\sigma_2, \sigma_3),\dots,(\sigma_n, \sigma_{n+1}) 
 \end{equation}
for the weight function $\omega_{ij}=\omega^+_{ij}-\tau \omega^-_{ij}$. Using the fact that any walk on a digraph with $M$ nodes can be decomposed into {the union of some} cycles and a path of length at most $M-1$, $\alpha(n)$ is decomposed as $\alpha(n)= \alpha_{*}(n)+\alpha_{2}(n)+\alpha_{3}(n)+\dots+\alpha_{M}(n)$. Here $\alpha_{*}(n)$ is the weight of the path and $\alpha_k(n)$ is the sum of the weights of all $k-$cycles. Note that the assumption $\tau > \nu (\mathcal G)$ implies $\omega(C):=\omega^+(C)-\tau\omega^-(C) < 0 $ for any cycle $C$. Namely, weights of all cycles are negative. Since $\mathcal V$ is finite, $\alpha_{ }(n)$ is bounded, and therefore, $\alpha(n) \rightarrow -\infty$ as $n \rightarrow \infty$. This is valid for all $\sigma\in\sddwell[\tau]$. Hence, {the} switched linear system (\ref{eq:nondefective_form}) is asymptotically stable.   
 \end{proof}
 
 \begin{remark}
{  If all subsystem matrices are non-defective, then Theorem \ref{th:nondefective_dwell} reduces to the discrete-time version of Theorem~1 in \cite{karabacak_scl_13}.}
 \end{remark}

\subsection{Bimodal case}
\label{sub:bimodal_dwell}
Theorem \ref{th:nondefective_dwell} can be enhanced for bimodal {non-defective} switched systems, namely switched systems with two {non-defective} subsystems. 
Since there is only one cycle in {the graph associated with} the bimodal case, the maximum cycle ratio is 
\begin{equation}
\label{eq:kappa}
\nu (\mathcal G)=\frac{\ln( \|P_2^{-1}P_1^{}\|\cdot\|P_1^{-1}P_2^{}\|)}{-\ln(\rho_1 \rho_2)}=\frac{\ln( \kappa(P_2^{-1}P_1^{}))}{-\ln(\rho_1\rho_2)},
\end{equation}
where $\kappa$ denotes the condition number for the spectral norm, namely $\kappa(A)= \|A\|\|A^{-1}\|.$
Therefore, using Theorem~\ref{th:nondefective_dwell}, the bimodal switched system is stable if
\begin{equation}
\label{eq:bimodal_tau}
\tau > \frac{\ln( \kappa(P_2^{-1}P_1^{}))}{-\ln(\rho_1\rho_2)}.
\end{equation}
It is known that eigenvectors can be scaled by any nonzero scalar.
%{(it is not true for generalized eigenvectors, thus non-defectiveness condition is imposed)}.
Then, an eigenvector matrix multiplied {on the} right by a nonsingular diagonal matrix is also an eigenvector matrix.  Let $\mathcal D$ denote the set of nonsingular diagonal matrices. Consider the eigenvector matrices $P_1$, $P_2$. Let $\bar P_1$, $\bar P_2$ be the new eigenvector matrices obtained {by} scaling columns of $P_1$, $P_2$ using $D_1$, $D_2 \in \mathcal D$, respectively. Then, we have
\begin{equation}
\bar P_2^{-1} \bar P_1^{}=D_2^{-1}P_2^{-1}P_1^{}D_1^{}.
\end{equation}
%Note that $\bar V_2^{-1} \bar V_1^{}$ is obtained by scaling rows and columns of $V_2^{-1} V_1^{}$
%Here, $S=V_2^{-1}V_1$, $D_L=D_2^{-1}$, $D_R=D_1$. 
Hence, the condition (\ref{eq:bimodal_tau}) in Theorem~\ref{th:nondefective_dwell} can be replaced by a stronger condition
\begin{equation}
\tau > \frac{\ln(\min_{D_L, D_R \in \mathcal D} \kappa(D_L^{}P_2^{-1}P_1^{}D_R^{}))}{-\ln(\rho_1\rho_2)}.
\end{equation}

There is no analytical method for minimizing the condition number for the spectral norm by scaling rows and columns, but algorithmic methods are available \cite{Morari_braatz,Liu_13}. 

\begin{corollary}
 \label{cor:bimodal1}
 The switched linear system (\ref{eq:nondefective_form}) with two Schur stable subsystems is asymptotically stable if
 \begin{equation}
\tau > \frac{\ln(\min_{D_L, D_R \in \mathcal D} \kappa(D_L^{}P_2^{-1}P_1^{}D_R^{}))}{-\ln(\rho_1\rho_2 )}.
 \end{equation}
 \end{corollary}

%Here, we use the two sided equilibration method in \cite{Liu_13} which is based on the idea that the condition number can be reduced by making norms of rows as well as norms of columns equal. 

One can use any other $p$-norm in Inequality (\ref{eq:same_as_jordan}). The non-defectiveness condition implies that $J$ is diagonal, hence $\|J\|_p=\rho(J)$, where $\rho $ denotes the spectral radius and $\|\cdot\|_p$ denotes the $p$-norm. Using the fact that $p$-norm is sub-multiplicative, one can similarly obtain the condition
\begin{equation}
\label{eq:bimodal_tau_remark}
\tau > \frac{\ln \kappa_p(P_2^{-1}P_1^{})}{-\ln(\rho_1\rho_2 )},
\end{equation} 
where, $ \kappa_p(A)=\|A\|_p\|A^{-1}\|_p $. There is an analytical method \cite{Bauer_63} for minimizing the condition number for norms $\|\cdot\|_1$ and $\|\cdot\|_{\infty}$ by scaling rows and columns. According to this method, for $p=1,\infty$,
\begin{equation}
\min_{D_L, D_R \in \mathcal D} \kappa_{p}(D_LAD_R)= \rho(|A||A^{-1}|)
\end{equation} 
where  $|A|$ denotes the matrix whose elements are absolute value of the corresponding elements of $A$, and $\rho$ denotes spectral radius.
Hence, we have the following result:

\begin{corollary}
 \label{cor:bimodal2}
 The switched linear system (\ref{eq:nondefective_form}) with $M=2$ is asymptotically stable if
 \begin{equation}
 \label{eq:bimodal2}
 \tau > \frac{\ln(\rho(|P_2^{-1}P_1^{}||P_1^{-1}P_2^{}|))}{-\ln(\rho_1 \rho_2 )}.
 \end{equation}
 \end{corollary}
 
%  \begin{remark}
%   If subsystem matrices are simultaneously triangularizable, then $\rho(|P_2^{-1}P_1^{}||P_1^{-1}P_2^{}|))=1$. Substituting this in (\ref{eq:bimodal2}), we get $\tau > 0$, which is the well-known result that switched linear systems with simultaneously triangularizable subsystem matrices are stable under arbitrary switching \cite{Mori_proc_97}.
%  \end{remark}

% \begin{remark}
% ************************************
% It is well-known that switched linear systems with simultaneously triangularizable subsystem matrices are stable under arbitrary switching \cite{Mori_proc_97}. Corollary~\ref{cor:bimodal2} can be used to give a simple proof of this result for bimodal systems. Let $A_1$ and $A_2$ be simultaneously triangularizable subsystem matrices of a bimodal switched system. Since simultaneous similarity transformation applied to $A_1$ and $A_2$ does not affect the condition (\ref{eq:bimodal2}), we can assume that $P_1$ and $P_2$ are triangular. Therefore, $S:=P_2^{-1}P_1^{}$ is also triangular. Let $s_i$ denote the $i$th diagonal element of $S$. Then, the $i$th diagonal element of the matrix $|S||S^{-1}|$ is equal to $|s_i|\cdot|\frac{1}{s_i}|=1$. This implies that all eigenvalues of $|S^{}||S^{-1}|$ are equal to one, since $|S^{}||S^{-1}|$ is triangular. Hence, we have $\rho(|S^{}||S^{-1}|)=\rho(|P_2^{-1}P_1^{}||P_1^{-1}P_2^{}|))=1$. Substituting this in (\ref{eq:bimodal2}), we get $\tau > 0$. 
% *****************************************
% \end{remark}

\section{Average Dwell Time Computation}
\label{sec:averagedwell}  
In this section, the average dwell time problem is considered, namely finding the smallest possible value $\tau$ for which the switched system (\ref{eq:basic_form}) is asymptotically stable for the average dwell time set (\ref{eq:avedwellset}). As a special case, the non-defective bimodal case is discussed in Subsection~\ref{sub:bimodal_average}. 

% \subsection{Non-defective case}
% For the case of non-defective subsystem matrices, a bound on the average dwell time can be found using the maximum cycle mean of the switching graph and the largest spectral radius $\rho_{\max}=\max_i{\rho_i}$ as follows: 
\begin{theorem}
 \label{th:nondefective_ave}
 Let $\{A_i\}_{i=1,\dots,M}$ be a family of non-defective Schur stable matrices and {let} $\mathcal G$ be a switching graph. 
 Then the switched linear system given by
\begin{equation}
x(t+1)=A_{\sigma(t)}x(t), \quad
\sigma \in \sdaverage[\tau,N_0], \quad t\in N
\end{equation}
is asymptotically stable for all $N_0$ if
 \begin{equation}
 \label{eq:max_ratio_bound}
 \tau >\frac{ \mu (\mathcal G)}{-\ln{\rho_{\max}}},
 \end{equation}
 where $\rho_{\max}=\max_i{\rho_i}$ and the maximum cycle mean $\mu$ is found using the weights given in (\ref{eq:omega_plus}) and (\ref{eq:omega_minus}) {with parameters} $\epsilon_i$ satisfying (\ref{eq:3}).
 \end{theorem}
 \begin{proof}
%  By the assumption that the switching signal has infinitely many switchings and {the} subsystems are non-defective, 
 For $t \in \{t_n, \dots, t_{n+1}-1\}$ we have the Inequality (\ref{eq:average_git}), which can be written as  
 \begin{equation}
 \|x(t)\|\leq\gamma e^{\alpha(n)+t\ln\rho_{\max}}\|x(0)\|,
 \end{equation} 
 where $\alpha(n)=\sum_{k=1}^n\ln\|P_{\sigma_{k+1}}^{-1}P_{\sigma_k}\|$. Consider the walk associated to the switching signal $\sigma$
 \begin{equation}
 \mathcal W_n := (\sigma_1, \sigma_2), (\sigma_2, \sigma_3),\dots,(\sigma_n, \sigma_{n+1}) 
 \end{equation}
 it is seen that the $\omega^+$-weight of the walk $\mathcal W_n $ is equal to $\alpha(n)$. Since a walk can be decomposed into the union of some cycles and a path of length less than $M$, where $M$ is the number of nodes, 
 \begin{equation}
 \alpha(n)=\alpha_*(n)+\alpha_2(n)+\alpha_3(n)+\dots+\alpha_M(n)
 \end{equation}
 where $\alpha_*(n)$ is the $\omega^+$-weight of the path of  length less than $M$, say $M_*$, and $\alpha_k(n)$ is the sum of the $\omega^+$-weights of all cycles of length $k$. Defining $\bar{\gamma}=\gamma e^{\max_{\mathcal W}\omega^+(\mathcal W)}$ where $\mathcal W$ {varies} over all possible paths, we get
 \begin{equation}
 \label{eq:alpha_ineq}
 \|x(t)\|\leq\bar{\gamma}e^{\alpha_2(n)+\alpha_3(n)+\dots+\alpha_M(n)+t\ln\rho_{\max} }\|x(0)\|.
 \end{equation} 
 Consider the maximum cycle mean of the switching graph $\mathcal G$, namely $\mu (\mathcal G)=\max_{C \in \mathcal C}\frac{\omega^+(C)}{|C|}$, where $|C|$ denotes the length of the cycle and $\mathcal C$ is the set of all cycles in $\mathcal G$. Then, it is obtained that $\omega^+(C)\leq |C|\mu (\mathcal G)$. Since $\alpha_2(n),\dots, \alpha_M(n)$ are cycle weights, we get
 \begin{eqnarray}
 \alpha_2(n)+\dots+\alpha_M (n)\leq (N_\sigma(t)-M_*)\mu (\mathcal G) \\
 \leq (N_0-M_*)\mu (\mathcal G)+\frac{t\cdot \mu (\mathcal G)}{\tau}.
 \end{eqnarray}
 Substituting this into (\ref{eq:alpha_ineq}) and defining $\bar{\bar{\gamma}}=\bar{\gamma}e^{N_0\mu (\mathcal G)}$, we obtain
 \begin{equation}
 \|x(t)\|\leq\bar{\bar{\gamma}} e^{\left(\frac{\mu(\mathcal G)}{\tau}+\ln\rho_{\max} \right)t}\|x(0)\|.
 \end{equation} 
 Since  $\frac{\mu(\mathcal G)}{\tau}+\ln{\rho_{\max}} < 0$ by assumption, we conclude that $\|x(t)\| \rightarrow 0$ 
 \end{proof}

 \begin{remark}
{  If all subsystem matrices are non-defective, then Theorem \ref{th:nondefective_ave} reduces to the discrete-time version of Theorem~2 in \cite{karabacak_scl_13}.}
 \end{remark}

%  
%  
%  \subsection{Defective case}
%  Let us generalize Theorem \ref{th:nondefective_ave} to any linear switched systems. For this purpose, generalized Jordan matrix decomposition of the form $A_i = P_i^{(\epsilon)}J_i^{(\epsilon)}{P_i^{(\epsilon)}}^{-1}$ can been used {as} defined in Section \ref{sub:defective}. Assuming that the $\epsilon_i$'s are chosen {satisfying} Inequality (\ref{eq:epsiloncondition}), define $\omega_{ij}^+=\ln\|{P_j^{(\epsilon_i)}}^{-1}P_i^{(\epsilon)}\|$. Rewriting Eq. (\ref{eq:defective_bound}) {under} the assumption that $\|J_i^{(\epsilon_i)}\|<1$, we have
%  \begin{equation}
%  \|x(t)\|\leq e^{\alpha(n) + t\ln J_{\max}}\|x(0)\|,
%  \end{equation} 
% where $\alpha(n)=\sum_{k=1}^n\ln\|P_{\sigma_{k+1}}^{-1}P_{\sigma_{k}}\| $ and $\|J_{\max}\| = \max_i\|J_i^{(\epsilon_i)}\|$. Hence, it is obtained that $\tau >\frac{ \mu (\mathcal G)}{-\ln{J_{\max}}}$. This shows that in the case of defective subsystem matrices {the} average dwell time can be computed by replacing $V$ and $\rho_{\max}$  in Theorem~\ref{th:nondefective_ave} with $P^{(\epsilon)}$ and $\|J_{\max}\|$ , respectively.
%  
 \subsection{Bimodal case}
 \label{sub:bimodal_average}
  {Similarly} to the minimum dwell time case, the average dwell time method can be improved for non-defective bimodal switched systems. As there is only one cycle in a bimodal system, $\mu (\mathcal G)= \frac{\omega^+(C)}{2}$. Then, the average dwell time {satisfies} $\tau>\frac{ \mu (\mathcal G)}{-\ln{\rho_{\max}}}=\frac{ \omega^+(C)}{-2\ln{\rho_{\max}}}=\frac{\ln \kappa_p(P_2^{-1}P_1^{})}{-2\ln{\rho_{\max}}}$. Hence, the method in Subsection \ref{sub:bimodal_dwell} can be applied to the computation of the average dwell time.

\section{Illustrative Examples}
We apply the obtained minimum dwell time computation method to two illustrative examples and compare the results with two different methods in {the} literature: The method  {given by Morse} in \cite{Morse_itac_96}, which finds a minimum dwell time guaranteeing each subsystem to be contractive, and the method {given by Geromel} \& Colaneri  in \cite{geromel_colaneri_06}, which uses linear matrix inequalities based on a multiple Lyapunov function technique. We skip the comparison of the obtained average dwell time computation {with} other methods in {the}  literature since {they} either require a specified convergence rate as in \cite{Zhang_09} or they refer to in mode-dependent form as in \cite{Zhang_13}, namely for each subsystem a certain average dwell time condition is imposed. 
\begin{example}
\label{ex1}
Consider a switched system consisting of four linear subsystems whose matrices are given by
\begin{equation}
 A_k=(U^{-1})^{k}\cdot A\cdot U^{k},\ \ k=0,\dots,3
 \label{eq:example}
\end{equation}
Here,

$$
A=\left(\begin{smallmatrix}
-0.2&1&0\\
-1&1.4&0\\
0&0&-0.4
\end{smallmatrix}\right)$$ 
and
$$U=\left(\begin{smallmatrix}
1.2&0&0\\
0&\cos(\frac{\pi}{3})&\sin(\frac{\pi}{3})\\
0&-\sin(\frac{\pi}{3})&\cos(\frac{\pi}{3})
\end{smallmatrix}\right).
$$
Assume that the switchings respect one of the switching graphs: fully connected $\mathcal G_1$ (Fig.~\ref{fig:fully}), one-sided ring $\mathcal G_2$ (Fig.~\ref{fig:onesided}) and two-sided ring $\mathcal G3$ (Fig.~\ref{fig:twosided}).
\end{example}

\begin{figure}
\centering
\includegraphics[width=0.4\textwidth]{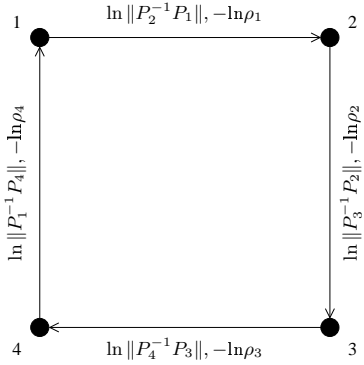}
\caption{The one-sided ring switching graph ($\mathcal G_2$) considered in Example~\ref{ex1}. On each edge, $\omega^+$ and $\omega^-$ are shown, respectively.}
\label{fig:onesided}
\end{figure}

For different switching graphs, minimum dwell time values are computed using  Theorem~\ref{th:nondefective_dwell} in Table~\ref{tab:ex1}. It can be seen that the results are better than the results obtained by the method given by Morse in \cite{Morse_itac_96}. Comparing the results with the method given by Geromel\&Colaneri in \cite{geromel_colaneri_06}, it can be seen that only for the switching graph $\mathcal G_3$, Theorem~\ref{th:nondefective_dwell} gives a worse result.

\begin{figure}
\centering
\includegraphics[width=0.4\textwidth]{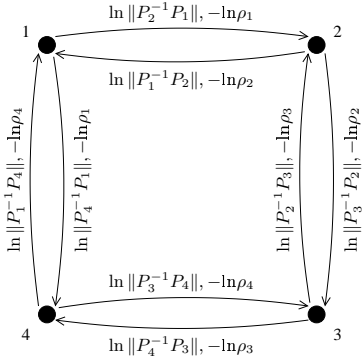}
\caption{The two-sided ring switching graph ($\mathcal G_3$) considered in Example~\ref{ex1}. On each edge, $\omega^+$ and $\omega^-$ are shown, respectively.}
\label{fig:twosided}
\end{figure}

Let us consider a bimodal system for which we can use Corollary~\ref{cor:bimodal1} and the  two-sided equilibration method in \cite{Liu_13} to compute a better value for the minimum dwell time than the one obtained by Theorem~\ref{th:nondefective_dwell}.  
\begin{example}
\label{example2}
Let $A_1$ and $A_2$ given below be the subsystem matrices of a switched system:
$$A_1=\left(\begin{smallmatrix}
-0.38&0.2&0.1\\
-0.16&0.72&0.16\\
-0.24&0.24&0.8
\end{smallmatrix}\right)
$$
$$A_2=\left(\begin{smallmatrix}
-0.8&-0.07&0.04\\
0.1&-1&0.05\\
-0.1&-0.06&-0.34
\end{smallmatrix}\right)
$$ 
$\tau$ is calculated as $7$ using the condition given in Theorem \ref{th:nondefective_dwell}. However, applying  Corollary \ref{cor:bimodal1}, $\tau$ is calculated as $1$ which is equal to the value found by linear matrix inequalities method \cite{geromel_colaneri_06}. Here, we use the two sided equilibration method in \cite{Liu_13} which is based on the idea that the condition number can be reduced by making norms of rows as well as norms of columns equal. The row scaling matrix $D_L$ and  the column scaling matrix $D_R$ are calculated as below:
$$D_R=\left(\begin{smallmatrix}
0.8528&0&0\\
0&0.7178&0\\
0&0&1.9789
\end{smallmatrix}\right)
$$ 

$$D_L=\left(\begin{smallmatrix}
1.8805&0&0\\
0&0.4700&0\\
0&0&1.8803
\end{smallmatrix}\right).$$ 
\end{example}

\section{Conclusion}
\label{sec:conc}
A method for the computation of the minimum dwell time that guarantees the asymptotic stability of a switched system has been presented. The method is applicable to systems whose switchings are governed by a digraph. The graph-theoretical nature of the method allows fast computation of an estimate of the minimum dwell time using the maximum cycle ratio algorithms in graph theory. We note that there are many problems that can be considered for the switched systems whose switchings are governed by digraphs. The role that the nature of the switching digraph plays on the dynamics of the switched system should be considered further.

We have shown that the average dwell time can be computed using the minimum cycle mean of the switching graph. This approach can be improved in two different ways: Firstly, one can introduce the mode-dependent average dwell time as in \cite{Zhang_13}, and try to find sufficient conditions on the mode-dependent average dwell times for a given switching graph. Secondly, one can consider a preassumed convergence rate as in \cite{Zhang_09} to calculate the average dwell time of a given switching graph in a less conservative method.

\section*{Acknowledgement}

This research was supported by T\"UB\.ITAK (The Scientific and Technological Research Council of Turkey), project no: 115E475.

\begin{table}[ht]
\centering
\caption{The minimum dwell time values computed for Example \ref{ex1}.}
	\begin{tabular}{|c|c|c|c|}
		\hline
		&\multicolumn{3}{|c|}{\textbf{\underline{Minimum dwell time}}}  \\
		{\textbf{Switching}}&& &Geromel\\
		{\textbf{Graph}}&Theorem~\ref{th:nondefective_dwell}&Morse\cite{Morse_itac_96}&Colaneri\cite{geromel_colaneri_06}\\
		&$\tau$&$\tau$&$\tau$\\
		\hline
		$\mathcal G_1$& 7  &8 &7 \\
		$\mathcal G_2$& 7 &8 &7 \\
		$\mathcal G_3$& 5 &8 &2 \\
		\hline
	\end{tabular}
\label{tab:ex1}
\end{table}

\bibliographystyle{vancouver} 
\bibliography{cycleweight}
\end{document}